\documentclass[conference]{IEEEtran}
\usepackage{amsmath}
\usepackage{algorithm}
\usepackage{enumitem}
\usepackage{algorithmic}
\usepackage{graphicx}
\usepackage[export]{adjustbox}
\usepackage{todonotes}
\usepackage{hyperref}
\usepackage{verbatim}
\usepackage{footmisc}
\usepackage{multirow}
\usepackage{amssymb,amsmath}
\usepackage{amsthm}
\usepackage{multirow,array}
\usepackage{url, hyperref}

\usepackage{breakurl}

\newtheorem{theorem}{Theorem}

\begin{document}

\title{Dynamic Curves for Decentralized Autonomous Cryptocurrency Exchanges}

\author{      Bhaskar Krishnamachari, Qi Feng, Eugenio Grippo\\
               Department of Electrical and Computer Engineering,\\ Viterbi School of Engineering \\ University of Southern California}
\date{\today}


\maketitle

\begin{abstract}
One of the exciting recent developments in decentralized finance (DeFi) has been the development of decentralized cryptocurrency exchanges that can autonomously handle conversion between different cryptocurrencies. Decentralized exchange protocols such as Uniswap, Curve and other types of Automated Market Makers (AMMs) maintain a liquidity pool (LP) of two or more assets constrained to maintain at all times a mathematical relation to each other, defined by a given function or curve. Examples of such functions are the constant-sum and constant-product AMMs. Existing systems however suffer from several challenges. They require external arbitrageurs to restore the price of tokens in the pool to match the market price. Such activities can potentially drain resources from the liquidity pool. In particular dramatic market price changes can result in low liquidity with respect to one or more of the assets and reduce the total value of the LP. We propose in this work a new approach to constructing the AMM by proposing the idea of \emph{dynamic curves}. It utilizes input from a market price oracle to modify the mathematical relationship between the assets so that the pool price continuously and automatically adjusts to be identical to the market price. This approach eliminates arbitrage opportunities and, as we show through simulations, maintains liquidity in the LP for all assets and the total value of the LP over a wide range of market prices. 
\end{abstract}

\section{Introduction}

Since the introduction of Bitcoin as the first peer-to-peer digital cash~\cite{nakamoto2019bitcoin}, the birth of different cryptocurrencies has revolutionized the world of finance~\cite{white2015market}. As of the time of writing this article, it is estimated that the total cryptocurrency market capitalization is more than \$600 Billion, involving thousands of different coins~\cite{coinmarket}. 

Traditionally, cryptocurrency exchanges, which use an order book mechanism, are centralized. They suffer from concerns about the concentration of financial power~\cite{hertzog2018bancor} and being prone to a single point of failure, resulting in a potentially significant loss of funds when attacked~\cite{angeris2019analysis}. Additionally, they also pose a liquidity problem for tokens with a smaller market capitalization resulting in barriers to entry to the financial market~\cite{hertzog2018bancor}. 

On the other hand, it is difficult to implement the order book model in a decentralized manner in the form of a blockchain smart contract~\cite{warren20170x}~\cite{angeris2019analysis}. First, market makers will face high gas costs to execute transactions, regardless of their sizes~\cite{warren20170x}. Second, it will require a complex matching algorithm to support a variety of order types~\cite{angeris2019analysis}.

Automated Market Makers such as Hanson's logarithmic market scoring rules (LMSRs) are widely used in traditional prediction markets to address the problem of low liquidity and trading volume~\cite{hanson2007logarithmic, wang2020automated}. An LMSR-based AMM is also used in decentralized prediction markets such as Gnosis~\cite{Gnosiswhitepaper2017} and Augur~\cite{peterson2015augur}. 
Given certain similar market characteristics, curve-based Automated Market Makers (AMM) were recently introduced to address the challenges in a currency exchange context. They are currently one of the areas of decentralized finance receiving the most attention. Instead of relying on the traditional market makers to provide liquidity,  decentralized exchanges utilizing curve-based AMMs, such as Bancor~\cite{hertzog2018bancor}, Uniswap~\cite{adams2020uniswap}, StableSwap/Curve~\cite{egorov2019stableswap} and many others implement a liquidity pool (LP) using smart contracts on a blockchain. In this model, liquidity providers supply single or multiple types of tokens to the designated liquidity pools, and traders exchange against the pools of tokens instead of relying on order matching. The liquidity pool of these AMMs track a pre-defined mathematical function (curve), thus determining how many tokens of one type to provide to a trader in exchange for a certain amount of another. Curve-based AMMs provide a continuous supply of liquidity compared to the order book model. Additionally, depending on the mathematical function (curve) utilized, they can potentially allow for a wide range of exchange prices. However, the token price within a liquidity pool for a given AMM (which we refer to as the pool price) might be different from the market price. 

When such a gap occurs on a decentralized AMM-based exchange, arbitrageurs may have the opportunity to buy or sell tokens to set the pool price equal to the market price, restoring equilibrium. However, in some cases, particularly when the market price changes dramatically, the AMM-based LP could lose liquidity with respect to one or more of the assets. We propose in this work a new approach to constructing the AMM by proposing the idea of \emph{dynamic curves}. It utilizes input from a market price oracle to modify the mathematical relationship between the assets so that the pool price continuously and automatically adjusts to be identical to the market price. This eliminates arbitrage opportunities and, as we show through simulations, helps the AMM-based LP maintain liquidity and total value over a wide range of market prices. 

The following are the key contributions of this work: 
\begin{itemize}
    \item We present a simple and unified mathematical and conceptual framework (in section~\ref{sec:background}) describing existing curve-based AMMs and key metrics such as pool price, slippage, divergence loss. It unifies much of what is known about them today. We believe this section will be of independent interest to researchers starting out in this area. 
    \item We focus on the liquidity problem posed by arbitrageurs on existing AMMs, especially when the market price for one of the assets becomes too high, causing asset depletion and value reduction of the liquidity pool. 
    \item We present a new dynamic curve mechanism, which is general enough to be adapted to any monotonic function/curve used on an AMM. This mechanism relies on an external market price oracle and eliminates arbitrageurs. We illustrate the mechanism concretely through generalizations of both constant-sum and constant-product models.  
    \item We present numerical simulations showing the clear advantages of our proposed dynamic curve mechanism in a) retaining greater liquidity in the pool to benefit small traders, b) retaining greater total value in the liquidity pool, and c) functioning effectively over a much larger range of market prices. 
\end{itemize}

The rest of this paper is organized as follows: we present and discuss relevant prior work in section~\ref{sec:related}. In section~\ref{sec:background} we give a unified treatment and definition of key concepts and metrics relevant to decentralized AMMs. In section~\ref{sec:dynamic} we propose and describe our new dynamic AMM models. In section~\ref{sec:sims} we present agent-based simulations and compare the performance of four different AMMs, including two static AMMs (Constant- Sum AMM, Constant-Product AMM) and their two dynamic generalizations that we introduce in this work. We present concluding comments in section~\ref{sec:conclusions}.

\section{Related Work \label{sec:related}}

Bancor was the first DEX to implement the type of AMM called Bonding Curve, which provides continuous liquidity \cite{hertzog2018bancor}. In this type of AMM, there is a single token (Bancor Network Token - BNT) used as an intermediate currency. There are separate pools for each non-native currency to be traded against BNT. This model is a little different from the arbitrary two-asset curve based AMM's that we focus on in this paper (though there are significant connections as well). In curve-based AMMs, any two currencies could be traded directly against each other. 

Borrowing solutions from the prediction market, Buterin~\cite{buterin_2016} first proposed such a curve-based AMM for a decentralized exchange. Specifically, he proposed the Constant Product Curve. It is a convex curve that takes the form of $x\cdot y=k$, where $x$ and $y$ are the total supply of two tokens in a liquidity pool and $k$ is the product constant. It was subsequently implemented by Adams \emph{et al.} \cite{adams2020uniswap} to create Uniswap. 

With the shape of a downward-sloping straight line, the Constant Sum Curve~\cite{berenzon_2020}\cite{chainlink_2020} takes the form of $x+y = k$. $x$ and $y$ are the total supply of two tokens in a liquidity pool, and $k$ is the sum constant. 
StableSwap/Curve \cite{egorov2019stableswap} implemented an AMM curve that is a blend of Constant Sum and Constant Product to provide continuous liquidity, price stability and a built-in pool balancing mechanism. 

Wang~\cite{wang2020automated} proposed the Constant Ellipse Curve AMM with the general form of $(x-a)^2 + (y-a)^2 + b\cdot xy = C$, in which $a$ and $b$ are constant. One can choose between the concave and the convex curve in the first quadrant~\cite{wang2020automated}. Wang also presents the curve corresponding to the LMSR rule.

Angeris and Chitra analyze such curve-based AMMs, which they refer to as constant function market makers in the general case, i.e., with arbitrarily many tokens~\cite{angeris2020improved}.  They analyze various mathematical properties of such AMMs, including formulating the optimal arbitrage by traders as a convex optimization problem. 

\subsection{Performance metrics for AMMs}

Slippage and divergence loss are the two main factors contributing to the proposal and adoption of different AMMs. The former is directly tied to the loss of traders, while the latter is directly connected to the liquidity providers’ returns. 

Slippage is the difference between the expected and actual trade execution price~\cite{totle_2019}, and in the AMM context, it is defined as the gap between the pool price before a trade and the effective price obtained for the trade (see section~\ref{sec:slippage}). As long as token price changes during trade, slippage incurs. In addition, when large trades happen compared to pool size, slippage increases dramatically, resulting in lower trading profits~\cite{chen_2019}. 

Divergence loss, sometimes called impermanent loss, incurs when liquidity providers withdraw liquidity with the presence of a difference in token price before and after trades~\cite{pintail_2020}. If funds are pulled out during a large price swing, liquidity providers will suffer a loss of total asset value, compared to simply holding the assets~\cite{chainsecurity}. Given trades might affect token pool prices and hence divergence loss, it is important to distinguish between regular and arbitrage trading. Divergence loss due to arbitrage trading in closing pool and market price gap can be mitigated by incorporating reliable oracles to protect liquidity providers~\cite{bancor_2020}.

\subsection{Slippage and Divergence Loss in AMMs}

In Bancor, given the dynamic pool price by design, trades experience slippage. Divergence loss incurs as the pool price is intrinsic and relies on arbitrageurs to close price gaps \cite{hertzog2018bancor}. To protect liquidity providers, Bancor v2 integrates with Chainlink price oracle to reduce divergence loss from arbitraging \cite{bancor_2020}.  


In Uniswap, similar to Bancor, the pool price is inherently unstable and the size of trades in relation to pool size affects pool price to different extents. The larger the trades are, the higher slippage and divergence loss can occur. 

The constant-sum curve has zero slippage \cite{berenzon_2020}\cite{chainlink_2020} and no divergence loss (as we show in section~\ref{sec:background}). However, because it has a fixed price and finite liquidity, it is only suitable for stablecoins and could easily be depleted of one of its pool assets; for this reason, it is primarily of theoretical interest \cite{berenzon_2020}\cite{chainlink_2020}. We use it as a baseline in our work.

StableSwap/Curve introduces an invariant that allows trading on a Constant Sum shaped curve when the portfolio is relatively balanced and switch trading to a Constant Product shaped curve when imbalanced \cite{egorov2019stableswap}. Such a design allows much lower slippage and divergence loss but is only applicable to stablecoins as the price of the desired trading range is always close to 1. 

The constant ellipse curve introduced by Wang~\cite{wang2020automated} has a fixed price range compared to that of a Constant Product Curve and thus a fixed range of slippage and divergence loss. Wang also concludes that the LMSR curve would not be suitable for exchanges if the numbers of the two tokens are not balanced in the liquidity pool. 

The proposal in this paper presents an approach for AMM-based decentralized exchanges using dynamic curves that eliminates the possibility of arbitrage and thus any divergence loss. Instead, as we show, depending on the chosen family of curves, any slippage loss incurred by traders is converted to an equivalent gain for the liquidity providers. A dynamic version of the constant sum curve is a special case of our proposed solution, and in that case, there is no slippage loss at all.  


\subsection{Simulation}
There are two simulations conducted respectively on StableSwap/Curve and Uniswap v1 to evaluate the DEXes performance. Egorov \cite{egorov2019stableswap} suggests that StableSwap/Curve generates 312\% APR and 0.06\% fee per trade for liquidity providers with total liquidity of \$30000 in DAI, USDC and USTD over 6 months. Angeris \emph{et al.} \cite{angeris2019analysis} conduct an agent-based simulation to test the hypothesis that Uniswap has a robust market mechanism with little arbitrage opportunities under various market conditions. Three types of agents, including profit-maximizing arbitrageurs, traders with exogenous motives and liquidity providers (both active/Markowitz portfolio optimizing and passive), interact in the Uniswap and a stochastic reference markets. The results show that Uniswap tracks market prices closely in different market environments, and Constant Product Curves have the potential to be price oracles \cite{angeris2019analysis}.

\section{Background - AMM Curves and Key Metrics \label{sec:background}}

Consider a liquidity pool with two coins, whose amounts are denoted by $x$ and $y$. For convenience, we will refer to these two tokens as $X$ and $Y$. The AMM will allow the exchange of one token for another following a given function $f$ as follows:

\begin{equation}
    y = f(x)
\end{equation}

We refer to a plot of this function showing all allowed combinations of $y$ and $x$ as the \emph{AMM curve}. For example, there can be a constant product curve, which is $y = \frac{k}{x}$ or a constant sum curve,  which would be denoted as $y = c - x$. It is generally considered reasonable for the AMM curve to be convex and monotonically decreasing because this ensures (as we shall see in the next section) that the price for the token $X$ is monotonically decreasing as a function of its availability in the pool, as should be expected of a typical supply curve.  

\subsection{Price and Bonding Curve}

Given an AMM curve, we can derive the price of the $X$ token as follows:

\begin{equation}
    p_X(x,y) = -\frac{dy}{dx}
\end{equation}

For example, for the constant product curve, we would get $    p_X(x,y) = \frac{k}{x^2}$ and likewise for the constant sum curve, we would get that $p_X(x,y) = 1 $.

A plot of $p_X(x,y)$ versus $x$ shows how the price of token $X$ varies with its supply in the liquidity pool. Such a curve is referred to as a \emph{bonding curve} or as a \emph{price curve}. Note that if $f(x)$ is monotonically decreasing, then the price will always be positive, and if $f(x)$ is convex, then the bonding curve will be monotonically decreasing (as it should, being a type of supply curve). 

\subsection{Value of the pool}
Given the definition of price, we can also assess the value of a given liquidity pool (measured in terms of $Y$) as follows: 

\begin{equation}
    V_{p}(x,y) = p_X \cdot x + y
\end{equation}

\subsection{Slippage \label{sec:slippage}}

\begin{figure}
    \centering
    \includegraphics[scale=0.25]{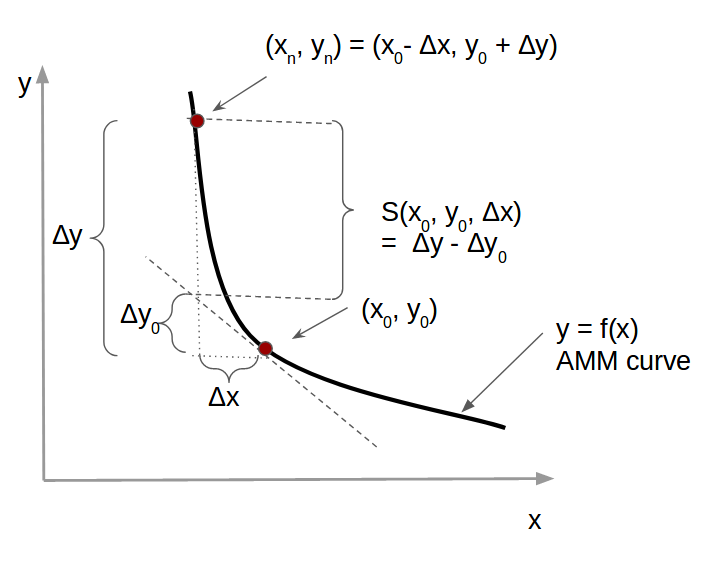}
   \caption{Illustration of Price Slippage on a Trade}
    \label{fig:slippage}
\end{figure}

For curve-based AMM, slippage is defined as the loss incurred by a trader due to the price mismatch between the pool price at which the trade is initiated and the effective price obtained during the trade. Let us consider a trader seeking to buy $\Delta x$ units of token $X$  when the LP is at a state $(x_0, y_0)$. Say that on the curve, the new point after the trade will be $(x_n, y_n)$, where $x_n = x_0 - \Delta x$. The amount that the trader would then need to put into the LP will be $
\Delta y = y_n - y_0$. If the pool price at the original point was 
$p_0$, then the buyer would have to pay $\Delta y_0 = p_0 \Delta x$. The difference between $\Delta y$ and $\Delta y_0$ is defined as the slippage loss $S(x_0,y_0,\Delta x)$. This is illustrated in figure~\ref{fig:slippage}. 

Similarly, when the trader wishes to sell $\Delta x$ units of token $X$, the gap between the $\Delta y_0 = p_0 \Delta x$ that the trader would like to receive and the $\Delta y = y_0 - y_n$ that he will actually receive would be the slippage loss on the sale, which could be expressed as $S(x_0,y_0,-\Delta x)$.

It is easy to see that on a constant sum AMM, the slippage loss is always 0 (because the price is constant at all points on the curve, or, equivalently, the tangent line at any point and the AMM curve always coincide). On any strictly convex curve, because the tangent line is always below the curve, the slippage loss will always be a positive quantity (i.e., the trader always incurs a penalty). The total slippage will be higher for a larger trade, and therefore acts as a disincentive for a trader to make large trades with the LP.






\subsection{Divergence Loss}

In general, when a trade is made, the price may change, as the original pair of values $(x_o,y_o)$ moves to a new pair $(x_n, y_n)$ following the curve, resulting in a new price $p_n$. Accordingly, the value of the liquidity pool could potentially decrease after a trade. This decrease as a relative or percentage decrease is referred to as divergence loss $\delta$, and can be formally defined as follows: 

\begin{eqnarray}
\delta & = & \frac{V_{p_n}(x_n,y_n) - V_{p_n}(x_0,y_0)}{V_{p_n}(x_0,y_0)}
\end{eqnarray}

We work out below the divergence loss for the two example curves.

\subsubsection{Divergence loss for constant-product curve}

For the constant product curve, recall that the following hold: 
\begin{eqnarray}
 p_n(x_n, y_n)  =  \frac{k}{x_n^2} & 
\implies x_n  =  \sqrt{\frac{k}{p_n}}  \\ 
 y_n  =  \frac{k}{x_n}  & 
\implies y_n  =  \sqrt{k \cdot p_n} 
\end{eqnarray}

Similarly, we also have that $x_o = \sqrt{\frac{k}{p_o}}$ and $y_o  =  \sqrt{k \cdot p_o} $. 

Then we can define $V_{p_n}(x_n,y_n)$ as follows:
\begin{eqnarray}
    V_{p_n}(x_n,y_n) &=&  p_n \cdot x_n  + y_n \nonumber \\
    & = & p_n \sqrt{\frac{k}{p_n}} + \sqrt{k\cdot p_n} \nonumber \\
    & = & 2 \sqrt{k \cdot p_n}
\end{eqnarray}
    
Likewise, we can define $V_{p_n}(x_o,y_o)$ as follows:
\begin{eqnarray}
    V_{p_n}(x_o,y_o) &=&  p_n \cdot x_o  + y_o \nonumber \\
    & = & p_n \sqrt{\frac{k}{p_o}} + \sqrt{k\cdot p_o} \nonumber \\
\end{eqnarray}

Based on the above two equations, we can calculate the divergence loss as follows: 
\begin{eqnarray}
    \delta & = & \frac{2 \sqrt{k \cdot p_n} - p_n \sqrt{\frac{k}{p_o}} + \sqrt{k\cdot p_o}}{p_n \sqrt{\frac{k}{p_o}} + \sqrt{k\cdot p_o}} 
\end{eqnarray}

Denoting by $\rho$ the ratio of the two prices $\frac{p_n}{p_o}$, the divergence loss for the constant product curve can be simplified to:
\begin{eqnarray}
    \delta &=& \frac{2\sqrt{\rho} - 1 - \rho}{1+\rho}
\end{eqnarray}

This result is given in~\cite{pintail_2020}.


\subsubsection{Divergence loss for constant-sum curve}

Here the price is always 1. The two values can be written as follows: 
\begin{eqnarray}
V_{p_n}(x_o,y_o) = x_o + y_o  &=& c \nonumber \\
V_{p_n}(x_n,y_n) = x_n + y_n    &=& c
\end{eqnarray} 

Since both are the same, the liquidity pool does not show any change in value, and thus the divergence loss in this case will be 0.

\section{Dynamic curves \label{sec:dynamic}}

In the prior work on AMMs, the curve has a fixed form and the exact shape is determined by the initial total liquidity. E.g., in the constant product curve, the parameter $k = x_i \cdot y_i$ where $x_i,y_i$ are the initial amounts of the two tokens. In other words, the curve can only change if the liquidity providers add/remove tokens from the pool, but not from trading activity. 

Consider a trade that happens while the market price of token X remains unchanged at some price $p_{mkt}$. If the pool changes from the state $(x_o,y_o)$ to a new state $(x_n,y_n)$, then the pool price would potentially change from $p(x_o,y_o)$ to $p(x_n,y_n)$  (assuming the curve is not the constant-sum curve in which case there is no change in the pool price). 
This can result in at least a temporary difference between the pool price and the market price. As we do in the rest of the paper, we are assuming here that the pool's capitalization is a relatively small fraction of the total market capitalization of the underlying assets so that the market price is not determined or affected by the pool price. 

Another reason for a temporary difference between the pool price and the market price could be that the market price changes due to some external market conditions.  In either case, traditionally, it is expected that these temporary differences will be erased by the action of arbitrageurs, restoring the pool price back to the market price.


We propose a new mechanism that instead changes the curve every time the market price changes in such a way as to ensure that the current pool price will always equal the market price, \emph{without requiring action by external arbitrageurs}. We illustrate below how this new mechanism would generalize the constant-product and constant-sum curves -- the same approach can be used to generalize other smooth, decreasing, convex curves to the dynamic setting as well.

\subsection{Dynamic curve adjustment to generalize constant-sum}

In this case, we can describe the market-price-tracking dynamic curve as follows:

\begin{equation} p_{mkt}(t) \cdot (x(t) - a(t)) + y(t) = c
\end{equation}

Here, the parameter $a(t)$ will also be adjusted dynamically when the market price changes, to ensure that the new linear curve passes through the current pair of $(x(t),y(t))$ values. For simplicity, say the market is initialized at some pair $(x(0),y(0)$ at a market price of 1. Then $c$ could be set to be  $x(0) + y(0)$, with the original $a(0)=0$. 

If the market shifts to a price of $p_{mkt}(t)$ at some time $t$ and the liquidity pool at this arbitrary time is $(x(t),y(t))$, then the value of $a(t)$ will also be adjusted as follows to match the above dynamic curve:
\begin{eqnarray}
a(t) = x(t) - \frac{c - y(t)}{p_{mkt}(t)} 
\end{eqnarray}

Intuitively, this dynamic curve is always a line that has the slope corresponding to the current market price and always passing through the current liqudity pair $(x(t),y(t))$. 

Any trade that happens uses the current (instantaneous) curve. This allows the constant sum AMM to flexibly support a wider range of market prices while still providing 0 slippage compared to the original design (which allows only a fixed pool price and thus will not work when the market price is dramatically different).

\subsection{Dynamic curve adjustment to generalize constant-product}

In this case, we can describe the market-price-tracking dynamic curve as follows:

\begin{equation} w(t) \cdot (x(t) - a(t)) \cdot y(t) = k \label{dynamic-cp-curve}
\end{equation}

Or alternatively, as:

\begin{equation}
    y(t) = \frac{\frac{k}{w(t)}}{x(t) - a(t)}
\end{equation}

Note that in the above expressions, $x(t)$ and $y(t)$ must always be strictly positive; $a(t)$ must be constrained to be always strictly less than $x(t)$; and $w(t)$ should always be strictly positive. The instantaneous price corresponding to the dynamic version of the constant product curve can be defined as follows: 

\begin{equation}
    p_X(t) = \frac{k}{w(t)} \cdot \frac{1}{(x-a(t))^2}  \label{price-dynamic-cp-curve}
\end{equation}

When the market price changes, then both $w(t)$ and $a(t)$ will have to be changed in order to (a) make sure that the new market price $p_{mkt}(t)$ matches $p_X(t)$ in equation~(\ref{price-dynamic-cp-curve}) and (b) x(t),y(t) match the curve described in equation~(\ref{dynamic-cp-curve}). Thus we have to solve two equations and two unknowns. The solution turns out to be the following : 
\begin{eqnarray}
a(t) & =& x(t) - \frac{y(t)}{p_{mkt}(t)} \nonumber \\
w(t) & =& \frac{k\cdot p_{mkt}(t)}{y(t)^2}
\end{eqnarray}

We remark: the first expression above ensures the requirement mentioned above that $a(t)$ will remain strictly less than $x(t)$ and the second expression ensures that $w(t)$ is strictly positive, so long as $k$, $p_{mkt}(t)$, $x(t)$ and $y(t)$ are all kept strictly positive at all.

\subsection{From divergence loss to slippage gain} 

As with the static setting, there is no slippage loss for traders or divergence loss in the case of dynamic constant sum AMM. This is because in the absence of any change in market price, the pool price does not change during a trade.

In the case of the dynamic constant product AMM, the traders do experience a slippage loss just like in the static constant product AMM case. However, corresponding gain in value is accrued entirely to the liquidity pool and could be referred to as a \emph{slippage gain} for the LP. Further, in the absence of change in the market price, because the pool price does not change in the dynamic constant product curve, there is no divergence loss. Rather, the LP benefits from each trade by the same slippage gain. Thus, the dynamic constant product AMM provides a strict improvement from the LP's perspective. This result, in fact, generalizes to the dynamic version of any strictly convex curve, as we show below.

\begin{theorem}
In a dynamic AMM based on a family of monotonically decreasing $y=f(x)$ curves that are strictly convex, when the market price remains fixed, the LP will gain value after each trade by an amount equivalent to the slippage loss of the trader.
\end{theorem}

\begin{proof} Assuming the market price does not change during a trade, for any strictly convex curve, the trader suffers a slippage loss at each trade. This is because the tangent to the curve (whose slope is equal to the pool price and therefore the market price) lies below the curve if it is strictly convex. If the trader buys $X$ tokens, it will therefore have to give the pool an amount of $Y$ tokens that exceeds what it should have given at the current market price. Likewise, if the trader sells $X$ tokens, it will receive an amount of $Y$ tokens less than what it should have received at the current market price. The gap, in either case, corresponds to the slippage loss. An equivalent amount is gained by the LP (when the trader buys $X$ tokens, the excess $Y$ tokens are sent to the LP; when the trader sells $X$ tokens, the gap corresponds to $Y$ tokens are withheld by the LP). There is no other source of divergence loss for the LP because the pool price is readjusted to the market price immediately after execution of the trade - its value increases precisely by the amount of slippage loss experienced by the trader.   
\end{proof}

\section{Simulation Results \label{sec:sims}}

We conduct a number of simulations to compare the performance of four different AMM's, two static AMM's (Constant-Sum AMM, Constant-Product AMM) and their two dynamic generalizations that we have introduced in section~\ref{sec:dynamic}.  Although it is not practically implemented due to its shortcomings, we nevertheless use the constant-sum AMM as a baseline to highlight how the dynamic version of this scheme has  advantages and may be more interesting from a practical perspective.

\begin{figure*}
    \centering
    \includegraphics[scale=0.12]{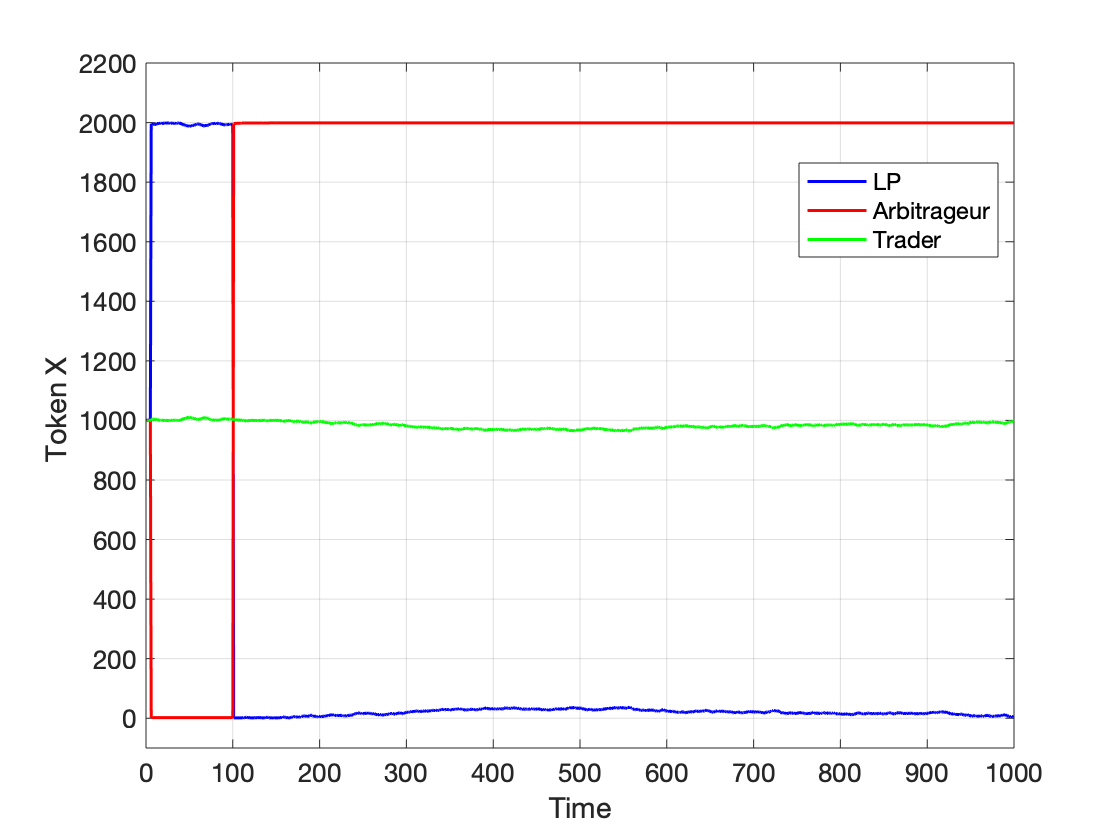} 
    \includegraphics[scale=0.12]{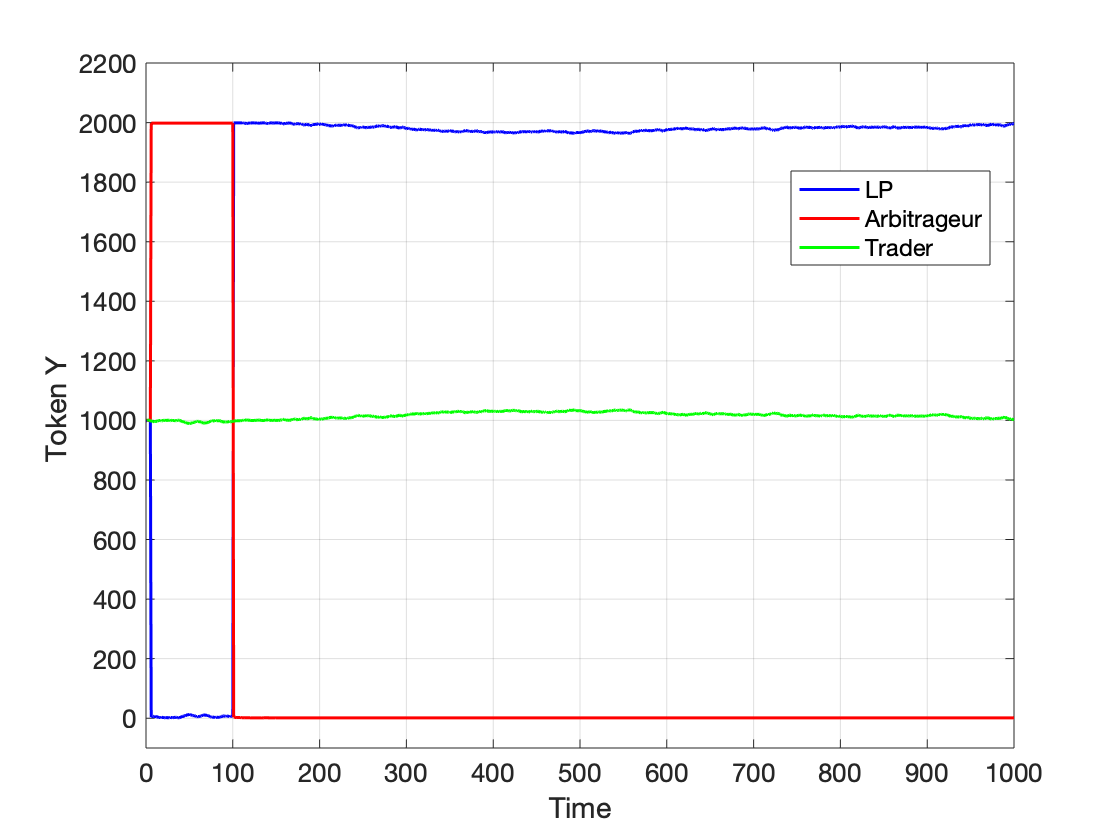}
    \includegraphics[scale=0.12]{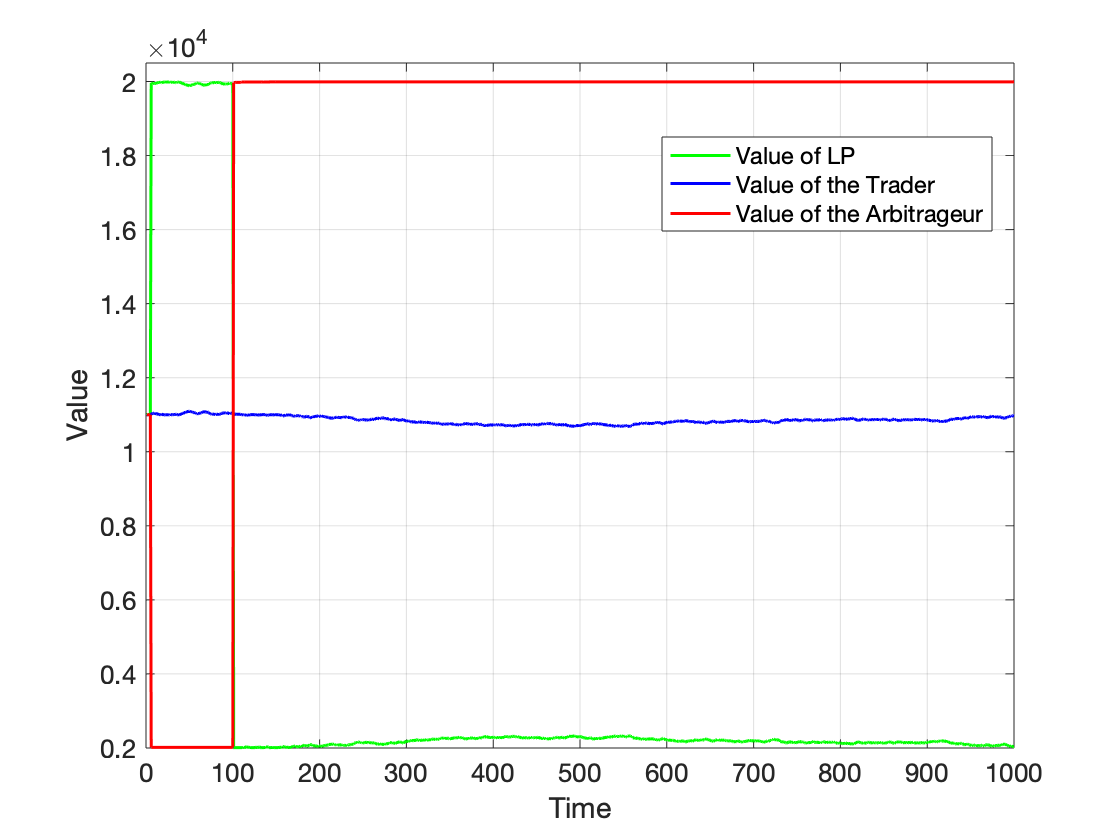}\\
    (a) X token holdings \hspace{0.8in}(b) Y token holdings \hspace{0.8in}(c) Total Value \\
    \includegraphics[scale=0.12]{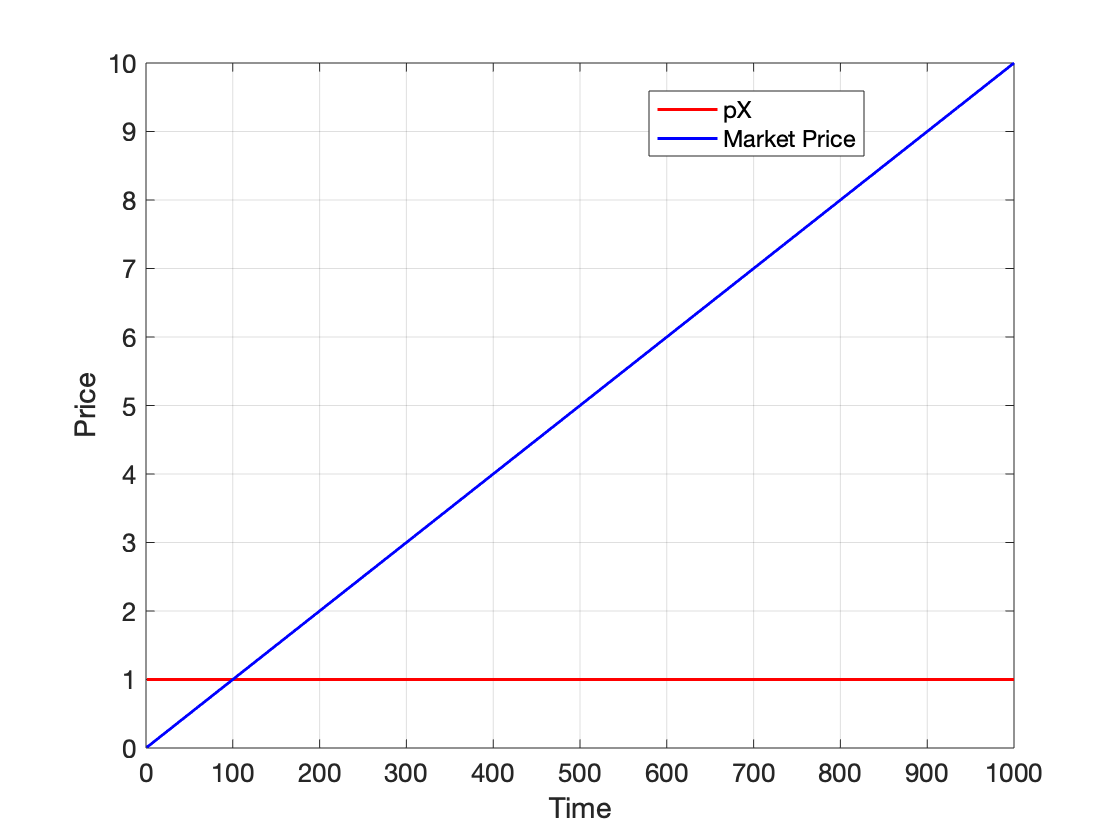}
    \includegraphics[scale=0.12]{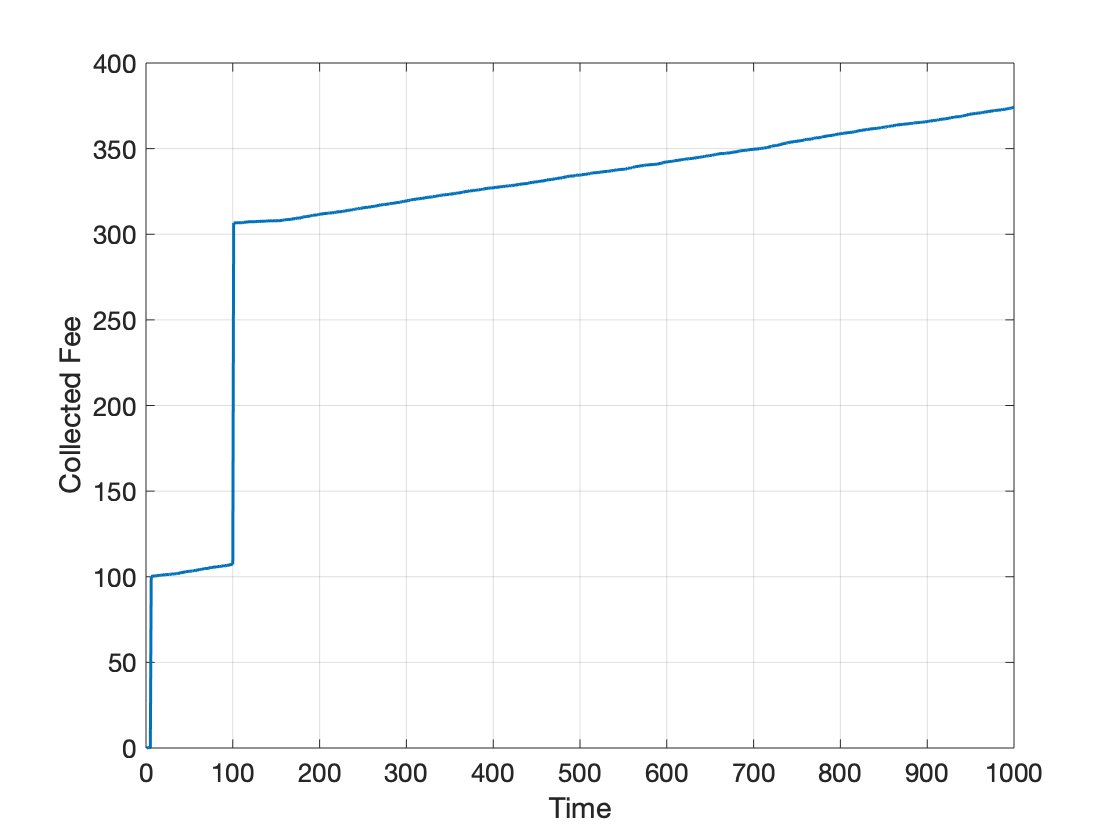}
    \includegraphics[scale=0.12]{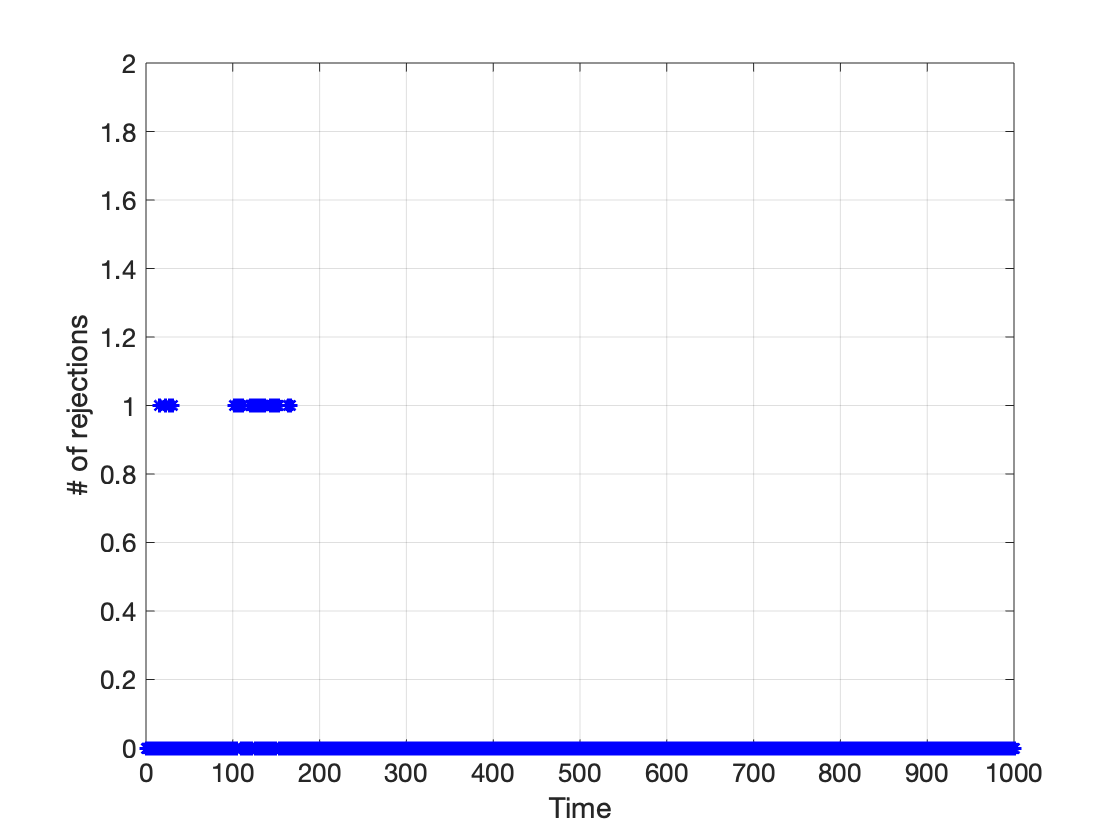}\\
        (d) Market vs Pool Price \hspace{0.75in}(e) Collected Fees \hspace{0.8in}(f) Declined Trades \\
    \caption{Simulation of a Static Constant Sum AMM}
       \label{fig:staticconstantsum}
 
\end{figure*}

\subsection{Simulation Setup}
In all our simulations, we consider the three parties involved: 
\begin{itemize}
    \item The liquidity pool (LP).
    \item A trader (which could be viewed as a collection or series of independent traders who are merely interested in exchanging a relatively small amount of one of the tokens in the LP for another).
    \item An arbitrageur (a profit-maximizing agent responding specifically to any gaps in price between the pool and the broader market; the arbitrageur in our simulation could also be viewed as a collection or series of independent arbitragers).
\end{itemize}
Initially, all three parties have 1000 $X$ tokens and 1000 $Y$ tokens. At each time, a trader makes a random trade (buy or sell $X$ token) drawn from a standard normal distribution (i.i.d. Gaussian process with zero mean and unit variance). The relatively small standard deviation would result in relatively light trading or profit from transaction fees collected from trades (that are not for the purpose of arbitrage).  We use the exact same sequence of random trades in evaluating all four AMM's, to enable a fair comparison. We model the market price as undergoing a linear increase from 0 to 10 over the course of 1000 time steps. 

\subsection{Static Constant Sum AMM}
Figure~\ref{fig:staticconstantsum} shows the results for the baseline static constant AMM. As shown by figure~\ref{fig:staticconstantsum}(d), the market price is initially below the pool price of 1, and after 100 time steps it switches
to being above the pool price. As can be seen from figure~\ref{fig:staticconstantsum} (a), (b), the arbitrageur initially sells all its $X$ tokens to buy up all of the $Y$ tokens from the LP when the market price is low, and then it buys up all the $X$ tokens from the LP when the market price becomes high. By draining the pool of its liquidity in one of its assets, to the pool's disadvantage, the arbitrageur effectively extracts most of the total value from the LP. Figure~\ref{fig:staticconstantsum}(c) shows the total value of the LP, the arbitrageur and the trader measured in terms of the market price at the end of the simulation (when $X$ tokens are worth 10 $Y$ tokens); it shows that the LP ends up with only about a tenth of the value it had at the beginning of the simulation. The LP does collect some fees from the arbitrageur (per figure\ref{fig:staticconstantsum}(e)), but they are relatively modest compared to the loss in value. Although the trader doesn't suffer slippage loss in this case, it can be disappointed and face a rejection of its requested trade whenever the pool lacks sufficient tokens of the type sought by the trader (as shown in figure~\ref{fig:staticconstantsum}(f); this happens particularly early on when the arbitrageur is able to cause the LP to be nearly empty of one of its assets, then another, disrupting the LP's ability to serve the regular trader).

\subsection{Static Constant Product AMM}

\begin{figure*}
    \centering
    \includegraphics[scale=0.12]{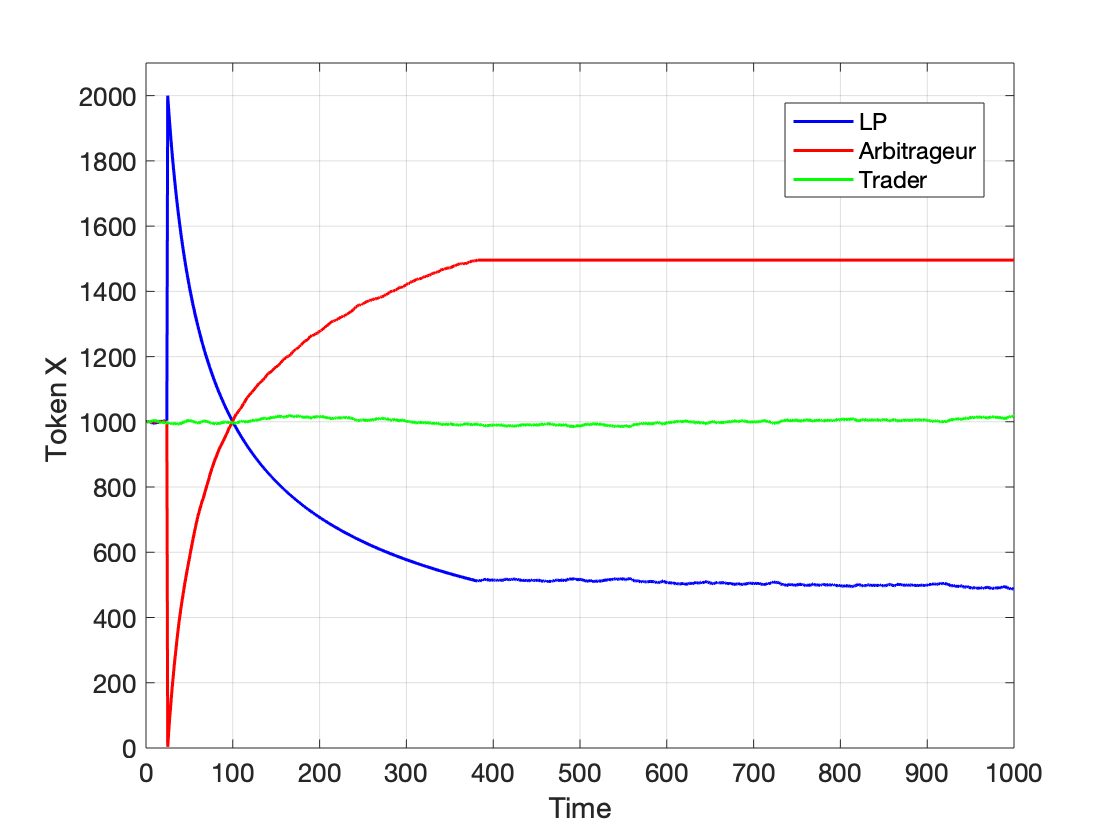}
    \includegraphics[scale=0.12]{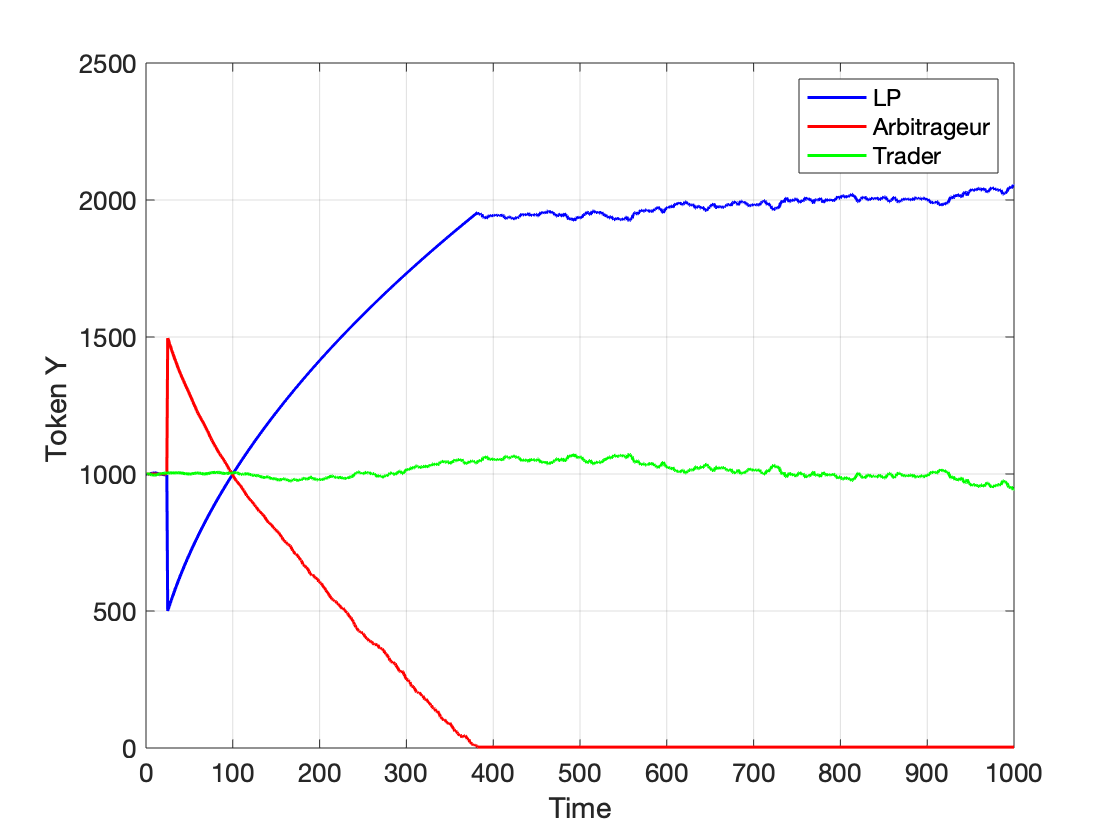}    
    \includegraphics[scale=0.12]{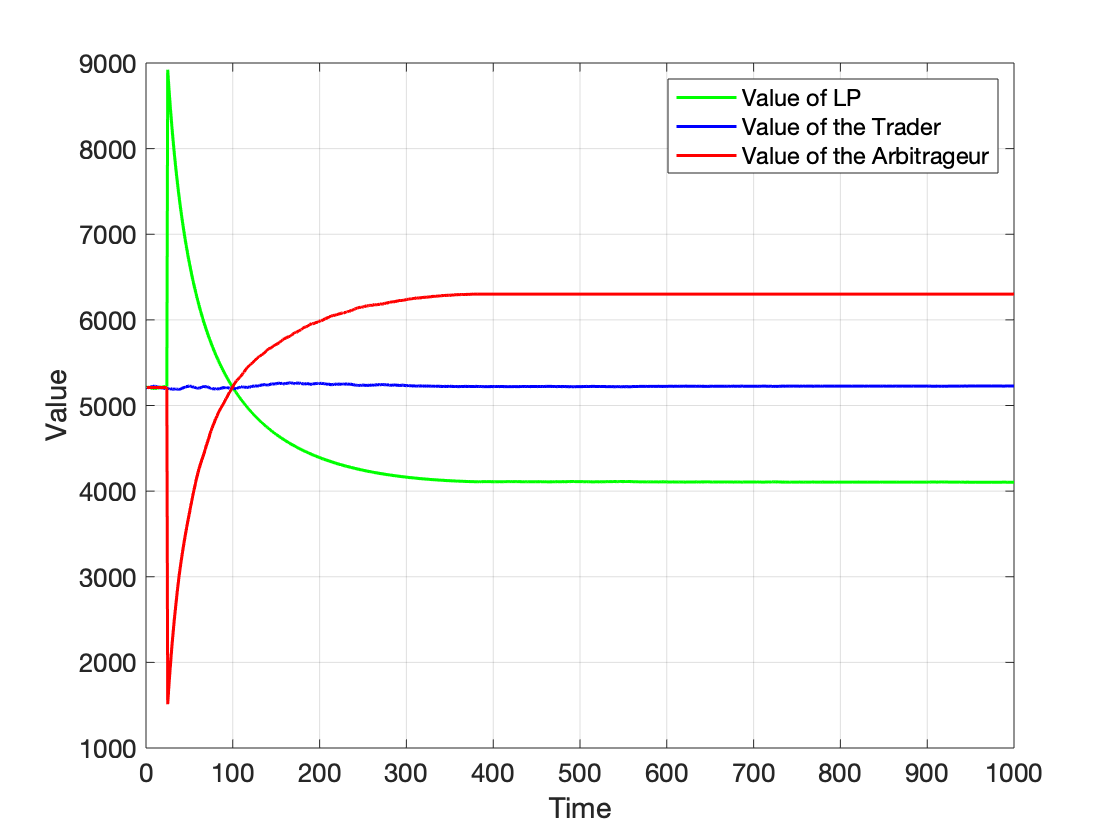}\\
    (a) X token holdings \hspace{0.8in}(b) Y token holdings \hspace{0.8in}(c) Total Value \\
   \includegraphics[scale=0.12]{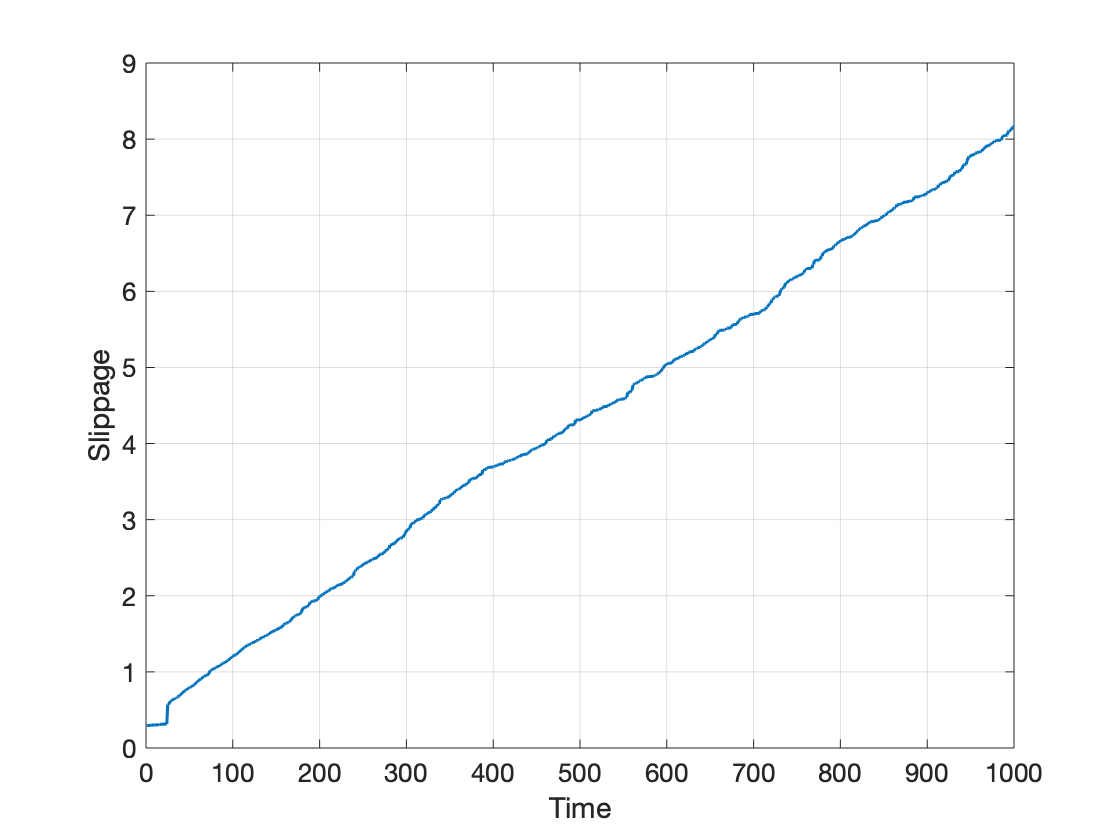}
    \includegraphics[scale=0.12]{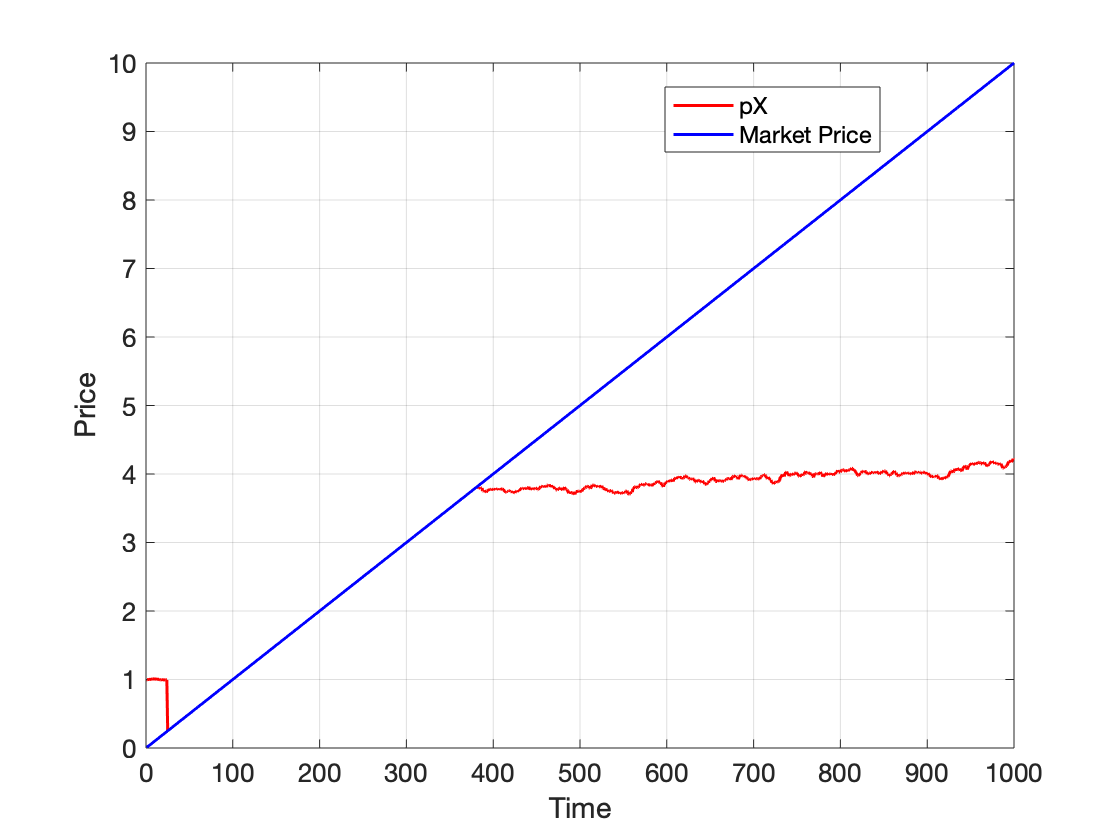}
    \includegraphics[scale=0.12]{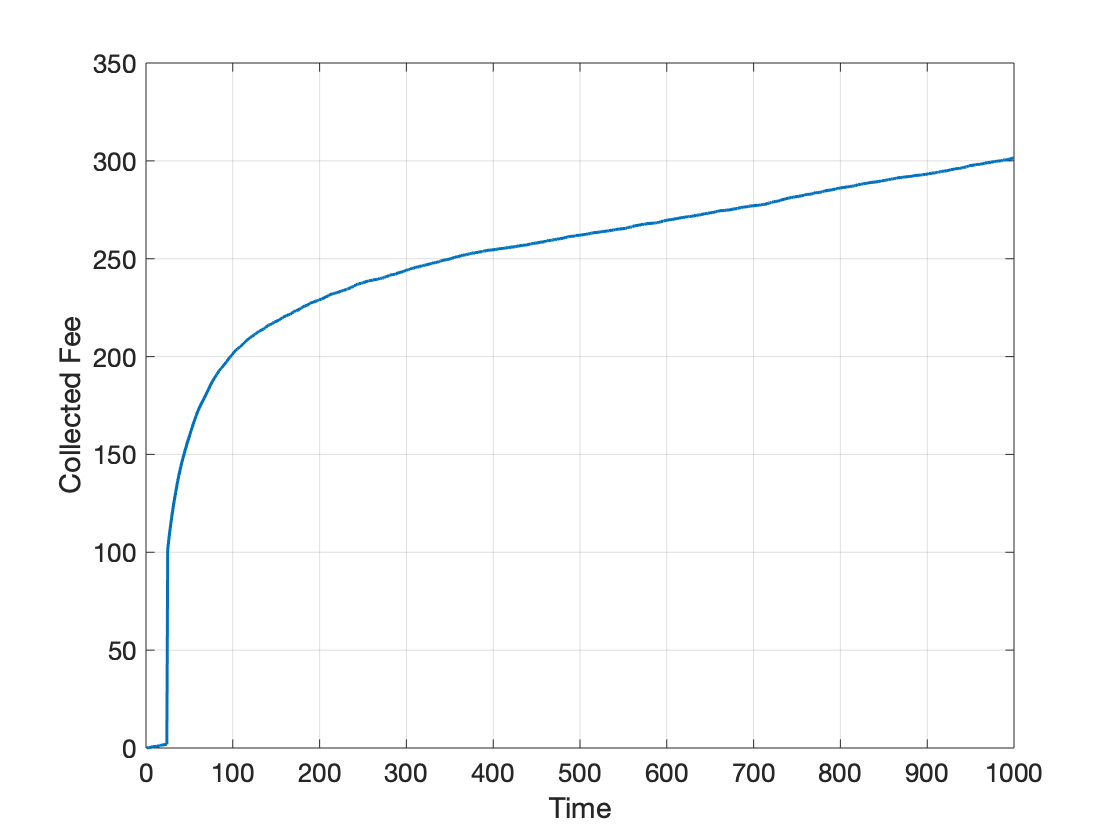}\\
        (d) Cumulative Slippage \hspace{0.65in}(e) Market vs Pool Price  \hspace{0.65in}(f) Collected Fees \\
    \caption{Simulation of a Static Constant Product AMM}
    \label{fig:staticconstantproduct}
\end{figure*}

\begin{figure*}
    \centering
    \includegraphics[scale=0.12]{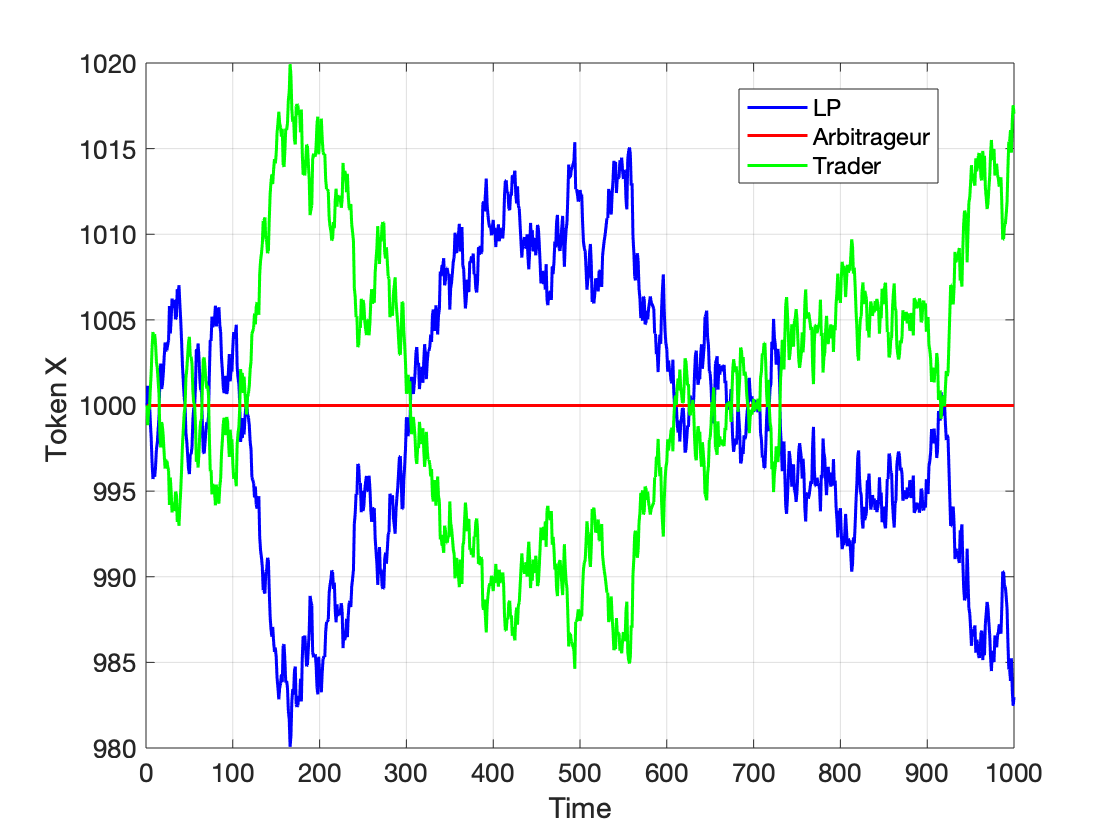}
    \includegraphics[scale=0.12]{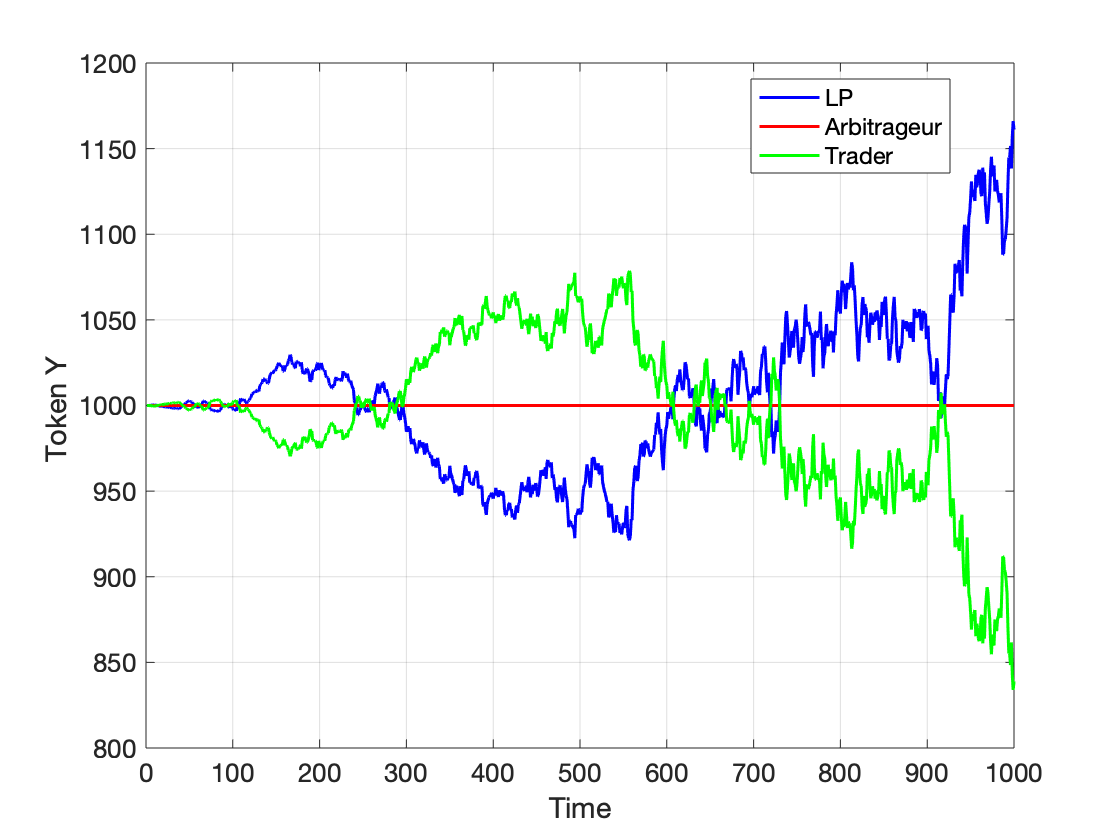}        
    \includegraphics[scale=0.12]{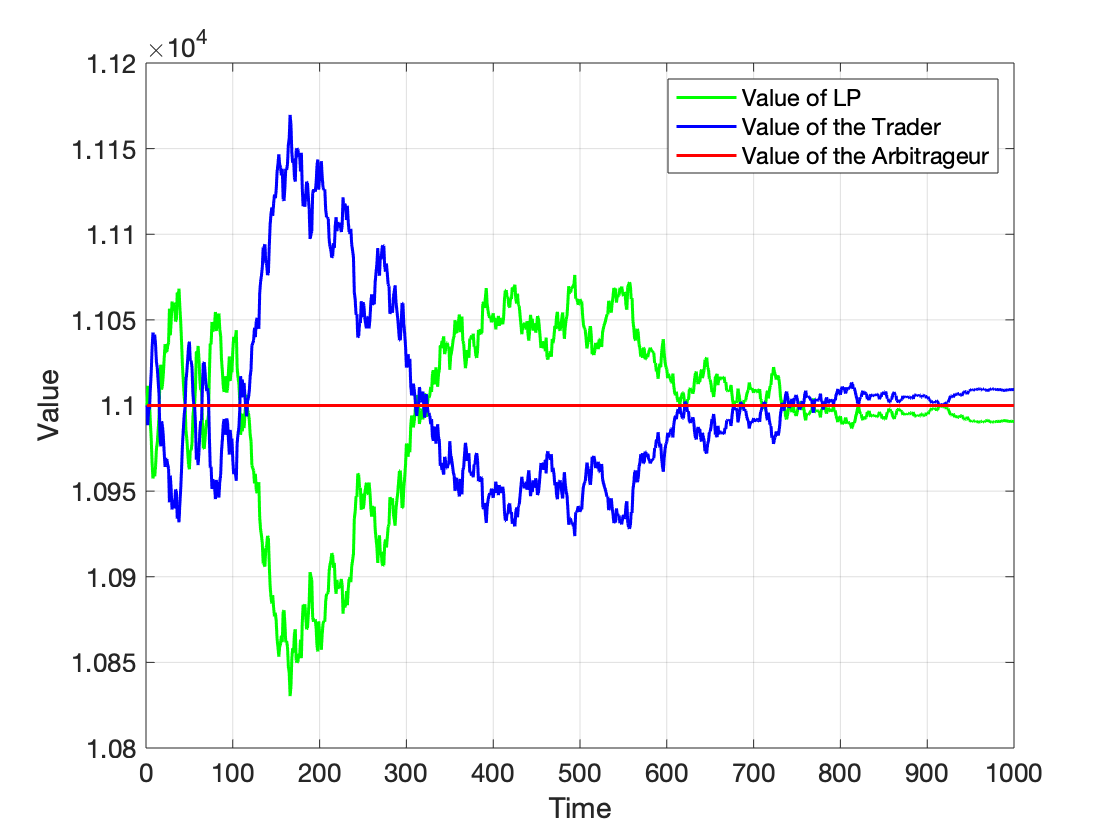}\\
    (a) X token holdings \hspace{0.8in}(b) Y token holdings \hspace{0.8in}(c) Total Value \\
    \includegraphics[scale=0.12]{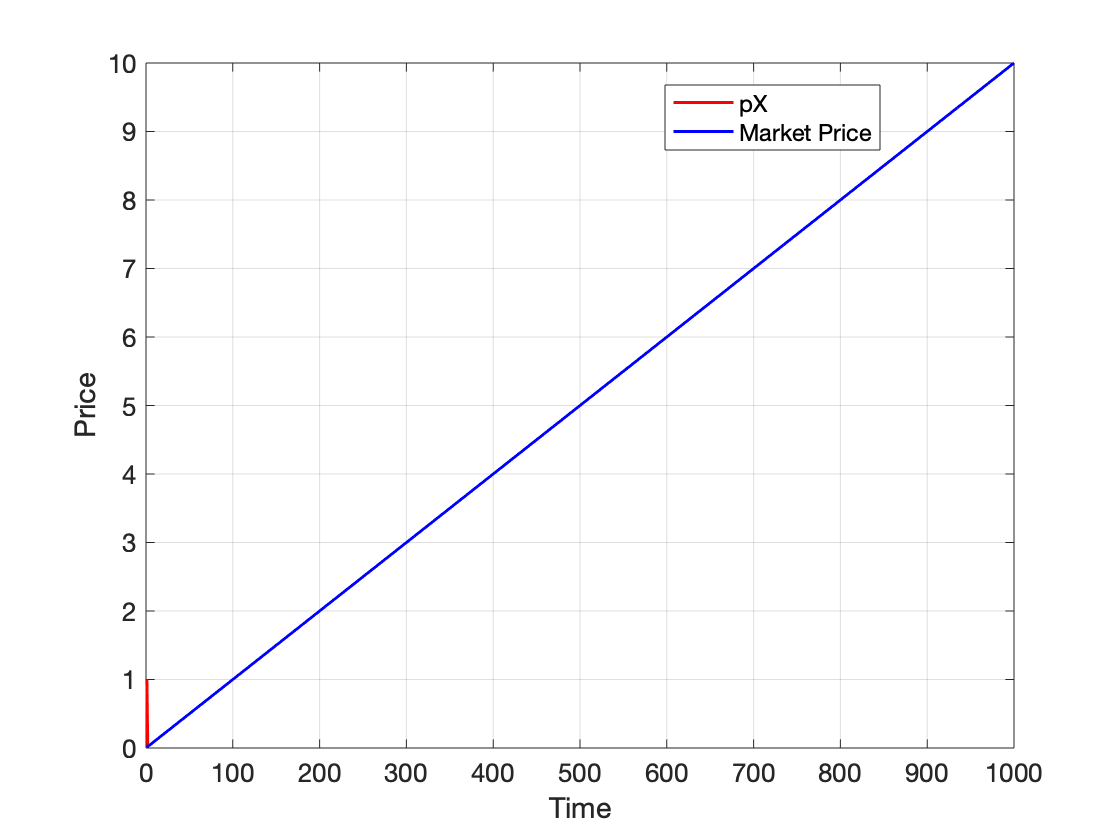}
    \includegraphics[scale=0.12]{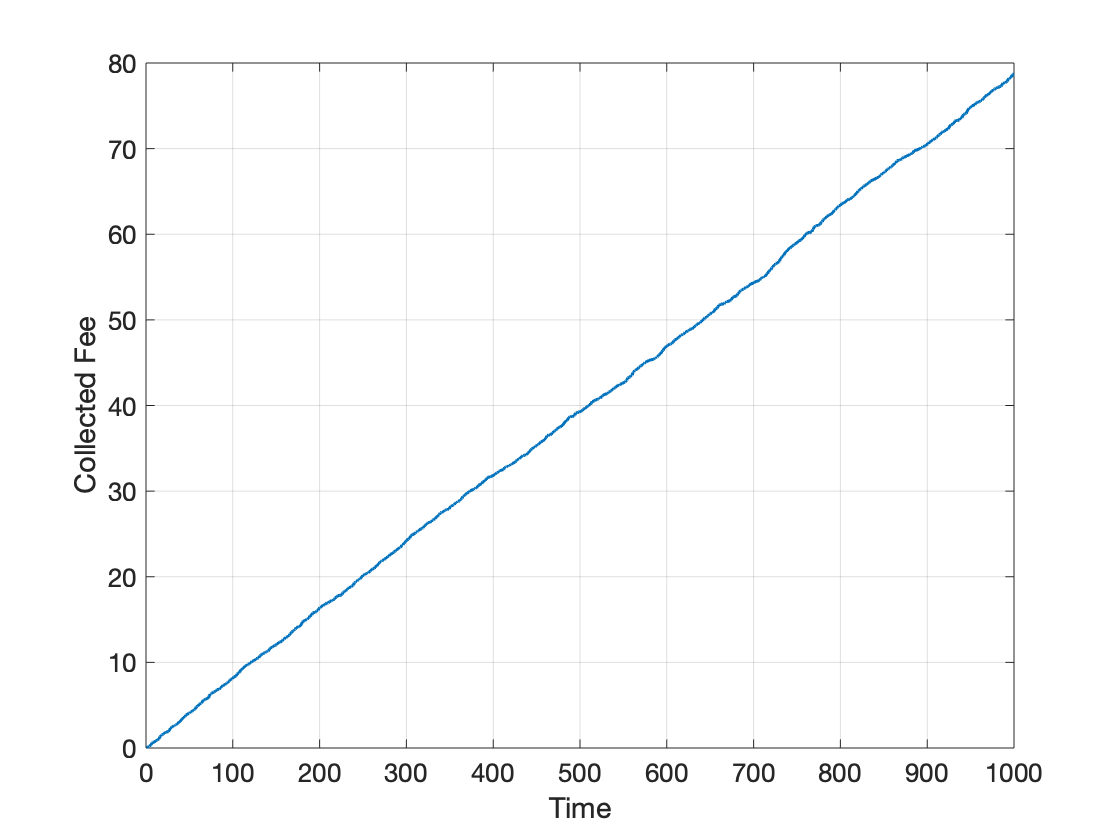}
    \includegraphics[scale=0.12]{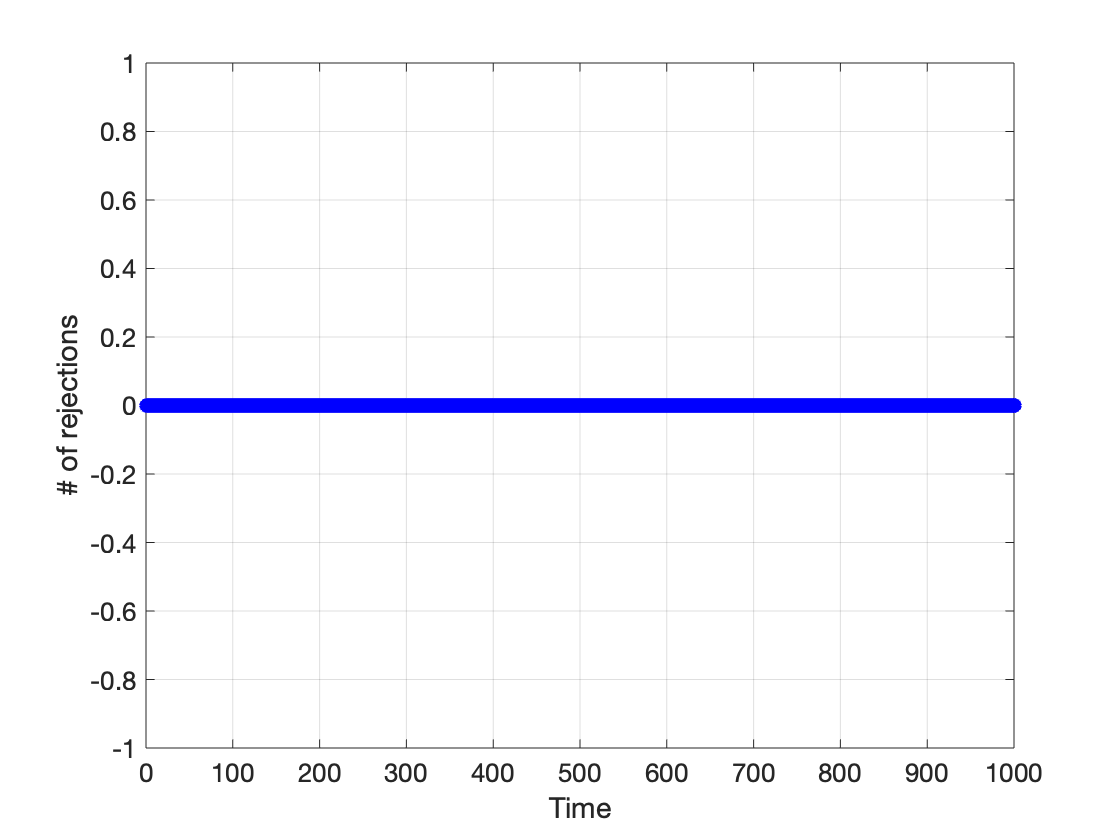}\\
        (d) Market vs Pool Price \hspace{0.8in}(e) Collected Fees \hspace{0.8in}(f) Declined Trades \\
    \caption{Simulation of a Dynamic Constant Sum AMM }
    \label{fig:dynamicconstantsum}
\end{figure*}

\begin{figure*}
    \centering
    \includegraphics[scale=0.12]{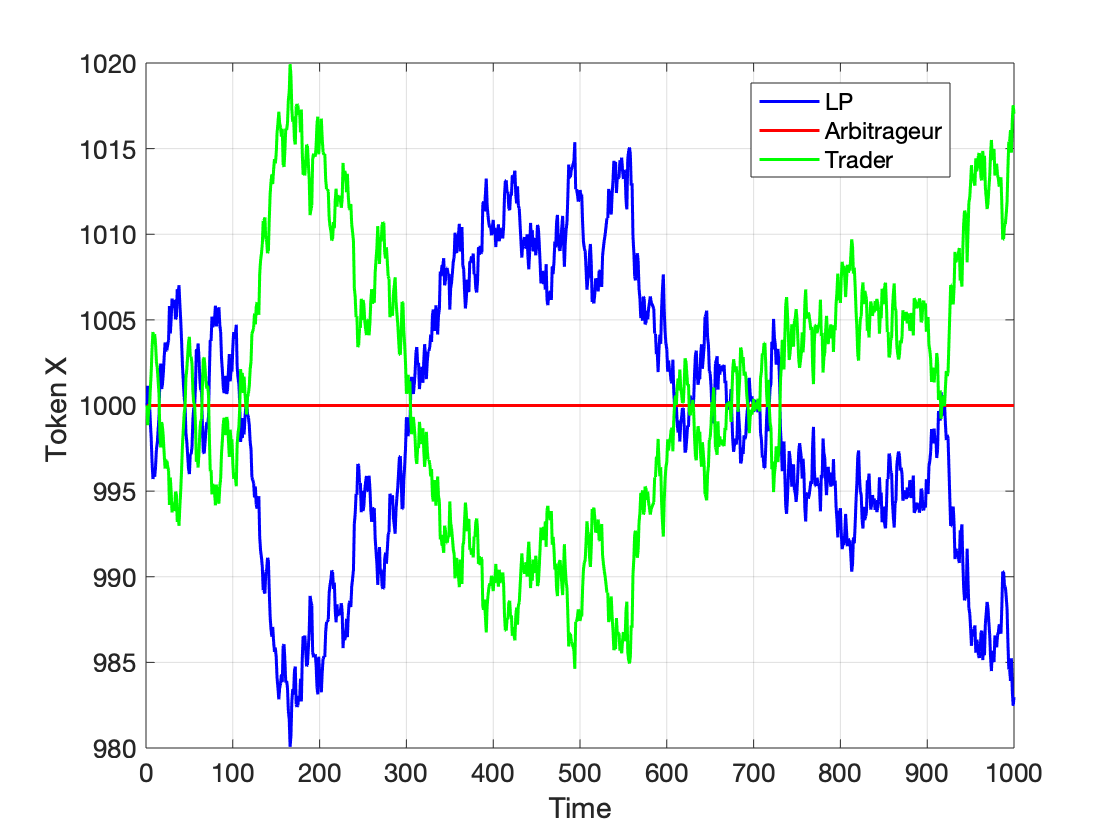}
    \includegraphics[scale=0.12]{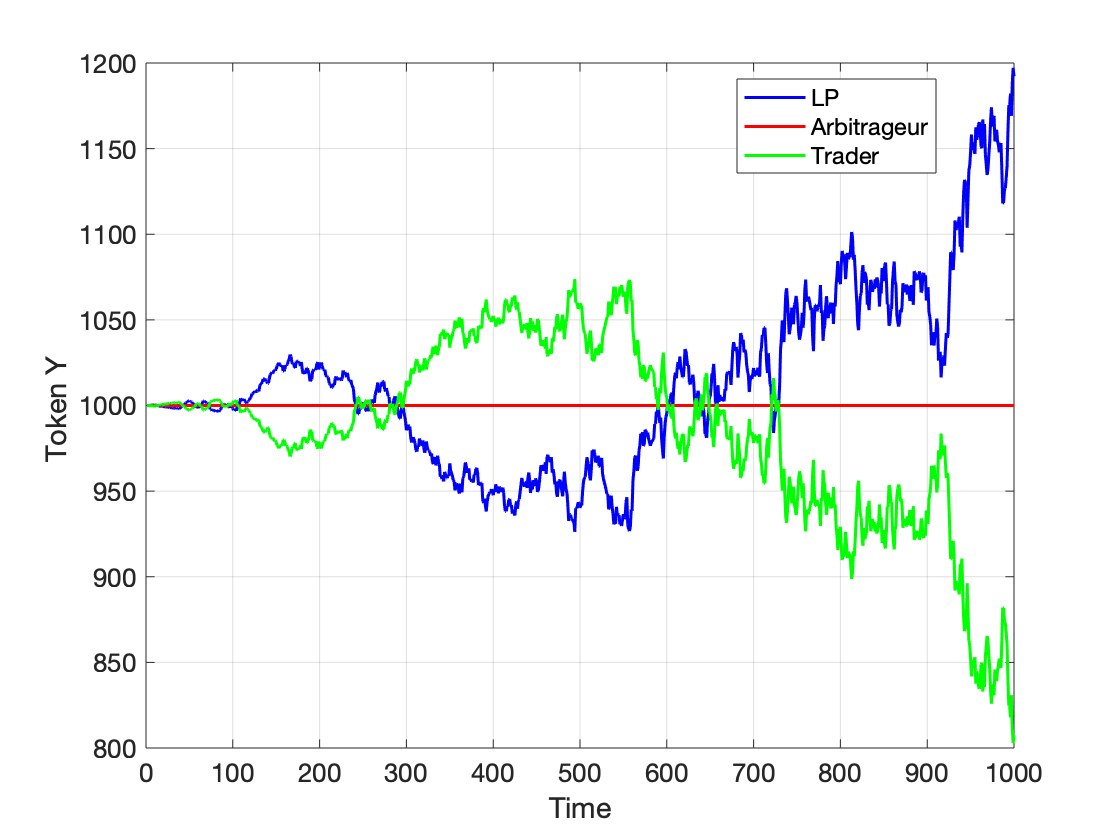}
    \includegraphics[scale=0.12]{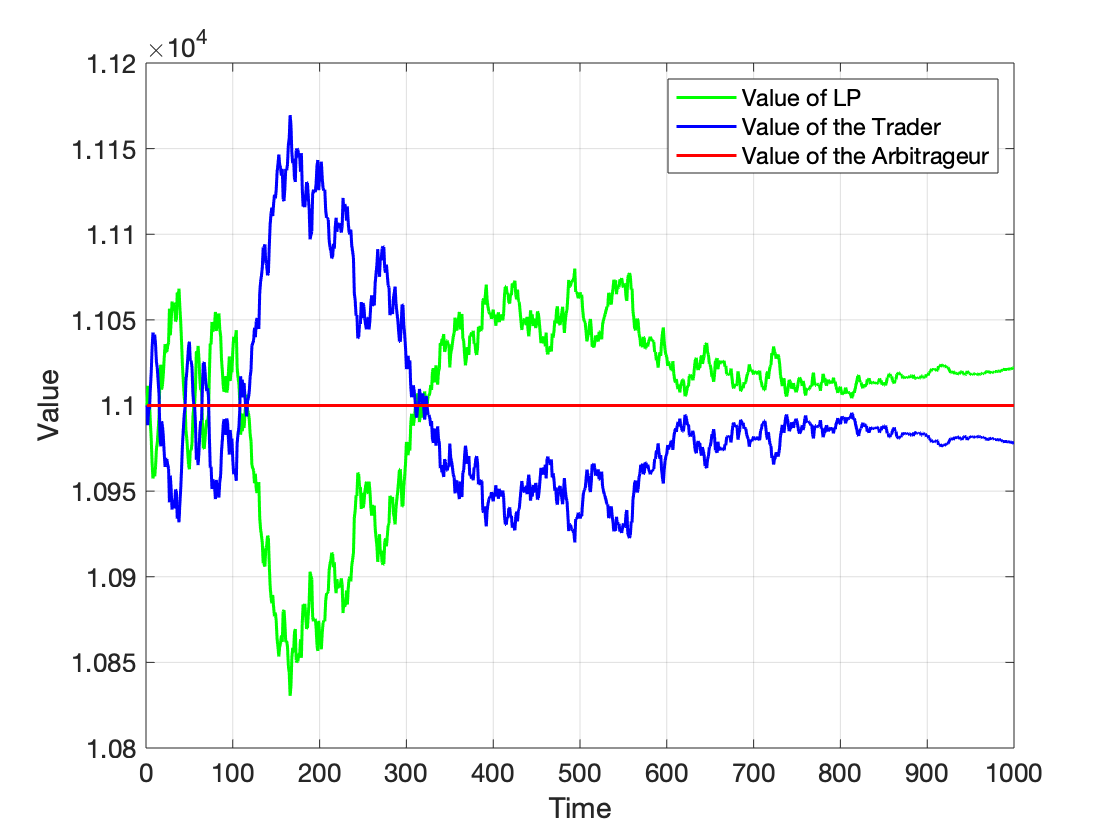}\\
    (a) X token holdings \hspace{0.8in}(b) Y token holdings \hspace{0.8in}(c) Total Value \\
    \includegraphics[scale=0.12]{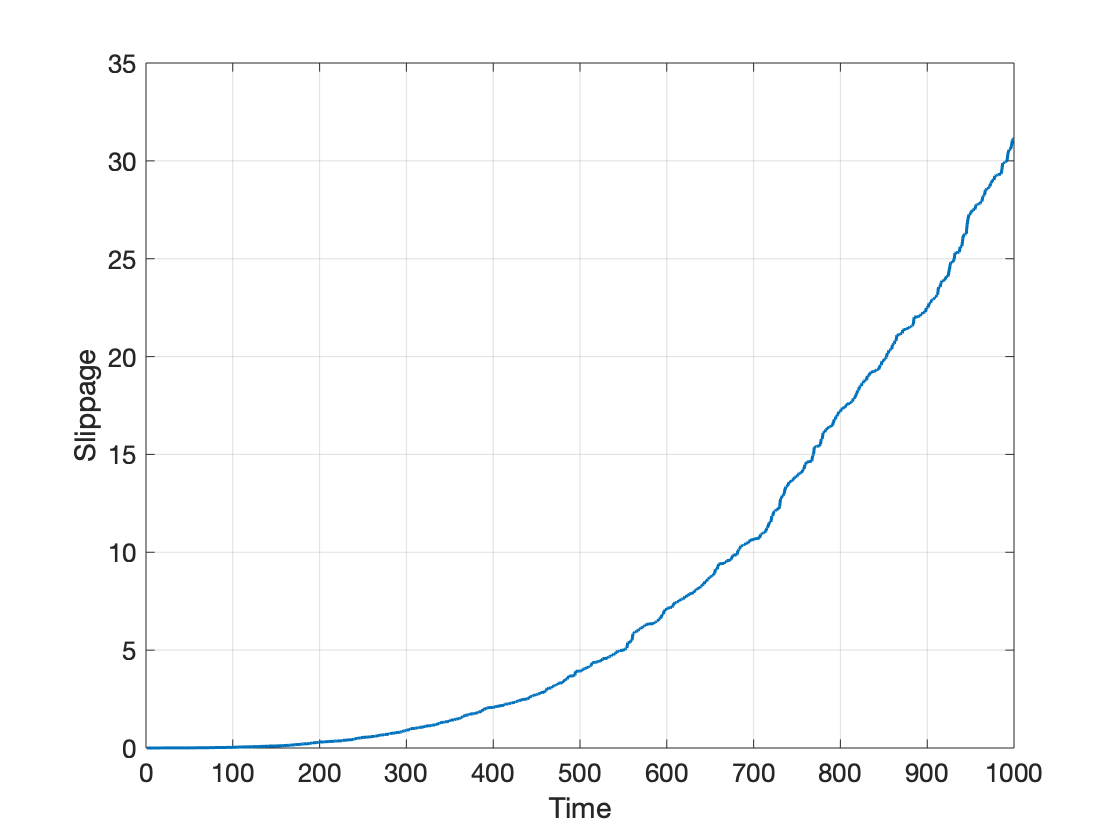}
    \includegraphics[scale=0.12]{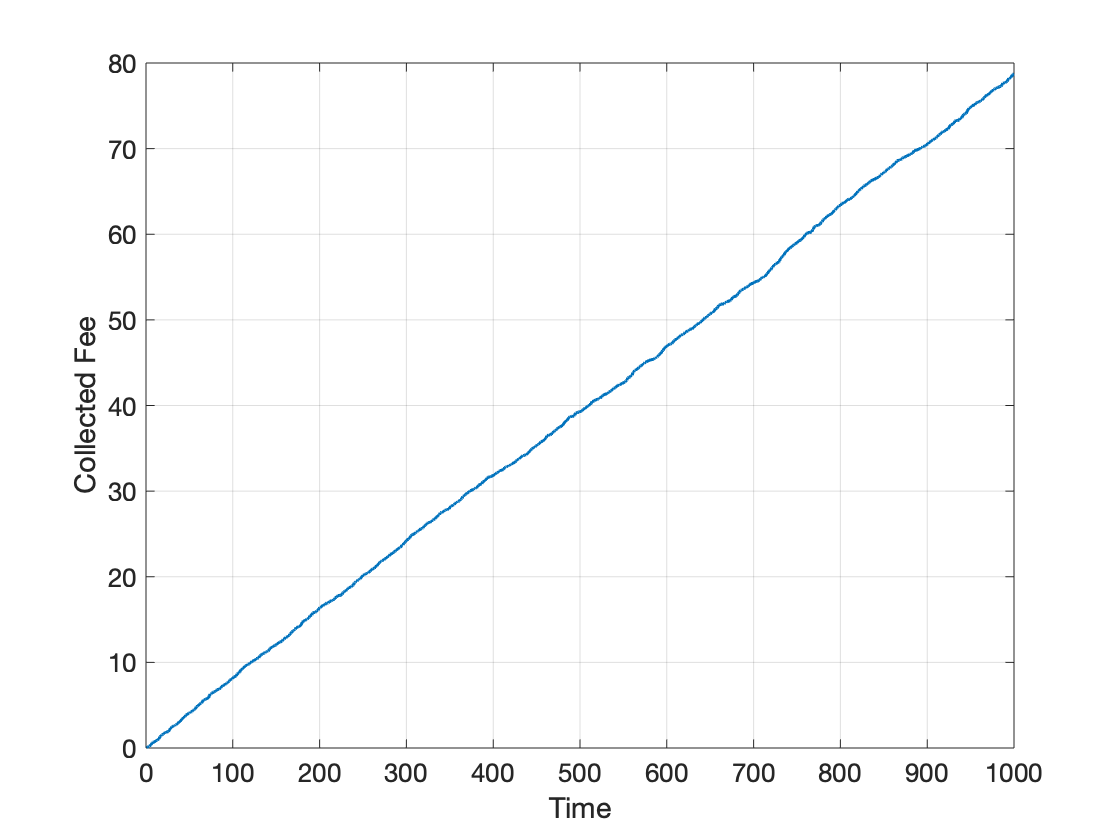}
    \includegraphics[scale=0.12]{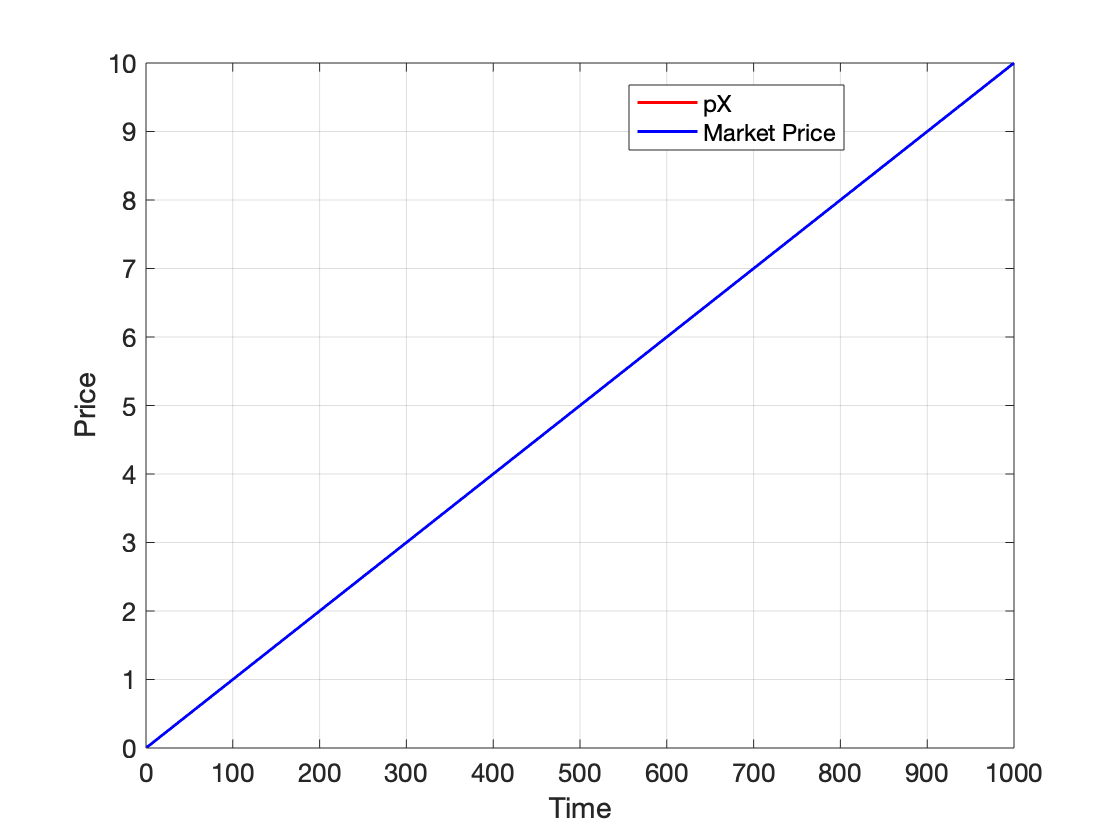}\\
        (d) Cumulative Slippage \hspace{0.65in}(e) Collected Fees \hspace{0.7in}(f) Market vs Pool Price \\

    \caption{Simulation of a Dynamic Constant Product AMM}
    \label{fig:dynamicconstantproduct}
\end{figure*}

Figure~\ref{fig:staticconstantproduct} shows the results for the static constant product AMM. Here too, we can see from Figures~\ref{fig:staticconstantproduct}(a) and (b), the arbitrageur has a big impact on the token holdings of the LP. Over time, as the market price goes up, the arbitrageur keeps buying $X$ tokens from the pool, equalizing the pool price and the market price. Eventually the arbitrageur has no more $Y$ tokens available (after which point the pool prize stays below the market price, see Figure~\ref{fig:staticconstantproduct}(e)). As shown in figure~\ref{fig:staticconstantproduct}(c), the total value of the LP reduces while that of the arbitrageur grows until the latter runs out of $Y$ tokens. In reality, however, a well-funded arbitrageur would keep going until all $X$ tokens are depleted from the LP, causing it to lose even more value. The trader suffers some constant slippage on average in every trade, as illustrated by the cumulative slippage showing a linear increasing trend  in Figure~\ref{fig:staticconstantproduct}(d).  There are again some collected fees from the arbitrageur and the regular trader, but here too they pale in comparison to the loss in value caused by the arbitrageur's actions. It is clear that just like with the constant sum setting, in the long term, the LP is at the mercy of the arbitrageur in the presence of significant market price volatility. 


\subsection{Dynamic Constant Sum AMM}

Figure~\ref{fig:dynamicconstantsum} shows the results for the dynamic constant sum AMM. Here the arbitrageur is eliminated as the pool price is automatically adjusted to the market price after each trade. The total tokens of either types $X$ or $Y$ remain the same and the holdings of the LP and the trader are mirror images of each other, as shown in figures~\ref{fig:dynamicconstantsum}(a) and \ref{fig:dynamicconstantsum}(b). The total value of the LP doesn't show a dramatic change other than random fluctuation caused by stochastic trading on the pool. In the absence of aggressive arbitrage, the collected fees are somewhat modest. As shown by figure~\ref{fig:dynamicconstantsum}(e), the LP collects the expected amount of fees (about 80 over 1000 trades, corresponding to the expectation of the absolute value of a normal random variable). A noticeable fact about the dynamic constant sum AMM is that \emph{the LP maintains more than 85\% of the original amount of tokens of both types at all times} during the simulation.  This is in sharp contrast to the liquidity problems observed in the static settings (figures~\ref{fig:staticconstantsum} and \ref{fig:staticconstantproduct}). As a result, there are no declined trades, as evidenced by figure~\ref{fig:dynamicconstantsum}(f). Another related comment to be made here is that there is little incentive for large trades in this dynamic AMM because the price is always pegged to the market price and the transaction fees (being a constant percentage) impose a higher cost in absolute terms  on larger trades.



\subsection{Dynamic Constant Product AMM}


Finally, figure~\ref{fig:dynamicconstantproduct} shows the results for the dynamic constant product AMM. Here too, the arbitrageur is eliminated as the pool price is automatically adjusted to the market price after each trade (see figure~\ref{fig:dynamicconstantproduct}(f)), and the $X$ and $Y$ tokens held by the trader and LP mirror each other as shown in figure~\ref{fig:dynamicconstantproduct}(a) and figure~\ref{fig:dynamicconstantproduct}(b). The total value of the LP in figure~\ref{fig:dynamicconstantproduct}(c) doesn't show a dramatic change other than the fluctuation caused by stochastic trading on the pool and a  slight increase in value for the LP due to the slippage gain discussed earlier in this paper. The corresponding cumulative slippage loss for the trader, equivalent to the cumulative slippage gain for the LP, is shown in figure~\ref{fig:dynamicconstantproduct}(d), and it is this gain that results in a slight positive drift in the LP value over time. As with the dynamic constant sum AMM, in the absence of aggressive arbitrage, the LP collects the expected amount of fees from the regular trader ($\approx 80$ over 1000 trades, corresponding to the expectation of the absolute value of a normal random variable). Compared to the dynamic constant sum setting, here, there is even less incentive for any trader to execute a large trade as they would suffer a large slippage loss in addition to paying transaction fees.  Figure~\ref{fig:dynamicconstantproduct}(a) and figure~\ref{fig:dynamicconstantproduct}(b) show that the LP maintains high liquidity in both assets again -- more than 90\% of the original amount of tokens of both types at all times during the simulation.

\section{Conclusions \label{sec:conclusions}}

We have given a detailed introduction to curve-based AMMs for decentralized cryptocurrency exchanges in this work. We introduced a new approach to operating such curve-based AMM decentralized exchanges that utilizes an oracle with a real-time market price feed to continuously and automatically adjust the pool price to the market price by dynamically adjusting the curves over time. We showed that in such a dynamic AMM, there is no room for arbitrage. The slippage loss for traders is converted to an equivalent gain for the liquidity pool. The LP maintains a high level of liquidity in all assets and its total value stays fairly stable over time. Such a dynamic AMM results in a slightly lower collection of transaction fees due to the elimination of arbitrageurs, but this loss is more than offset by the ability to maintain the total value of the pool. 

From a practical perspective, implementing such a dynamic AMM requires the use of a low-latency and accurate market price oracle. It would be of great interest to develop and evaluate a real-world decentralized exchange based on this approach.

\bibliographystyle{ieeetr}
\bibliography{references}

\end{document}